\documentclass[11pt]{article}
\usepackage{graphicx} 
\usepackage[margin=1in]{geometry} 
\usepackage{setspace}
\usepackage{amsmath,amsfonts,amsthm,amssymb,xcolor,mathtools}
\usepackage{comment, algorithm, algpseudocode}
\usepackage[style=alphabetic,natbib=true]{biblatex}
\usepackage{tikz}
\usepackage{dsfont}

\newcommand{\bb}{\mathbb}

\usepackage{hyperref}
\hypersetup{
    colorlinks,
    linkcolor={blue},
    citecolor={green},
    urlcolor={red}
}

\newtheorem{theorem}{Theorem}[section]
\newtheorem{lemma}[theorem]{Lemma}
\newtheorem{corollary}[theorem]{Corollary}
\newtheorem{claim}[theorem]{Claim}

\newtheorem{definition}[theorem]{Definition}

\title{\Large{Concentration of Submodular Functions and Read-$k$ Families Under Negative Dependence}}
\date{}
\author{\normalsize{Sharmila Duppala, George Z. Li, Juan Luque, Aravind Srinivasan, and Renata Valieva}}
\date{\normalsize{University of Maryland, College Park}}

\begin{document}

\maketitle
\begin{abstract}
    We study the question of whether submodular functions of random variables satisfying various notions of negative dependence satisfy Chernoff-like concentration inequalities. We prove such a concentration inequality for the lower tail when the random variables satisfy negative association or negative regression, partially resolving an open problem raised in (\citet{approx/QiuS22}). Previous work showed such concentration results for random variables that come from specific dependent-rounding algorithms (\citet{focs/ChekuriVZ10,soda/HarveyO14}). We discuss some applications of our results to combinatorial optimization and beyond. We also show applications to the concentration of read-$k$ families \cite{rsa/GavinskyLSS15} under certain forms of negative dependence; we further show a simplified proof of the entropy-method approach of \cite{rsa/GavinskyLSS15}. 
\end{abstract}
\section{Introduction}

Concentration inequalities are ubiquitous in discrete mathematics and theoretical computer science~\cite{probmethod,dubhashipanconesi}. The most canonical examples are the Chernoff-Hoeffding bounds, which show strong concentration for linear combinations of independent random variables~\cite{chernoffbound,hoeffdingbound}. In some applications, the condition of independence is too restrictive, so weaker notions have been considered~\cite{azuma,siamdm/SchmidtSS95,approx/Skorski22}. Of interest to us is the setting where the random variables are negatively correlated, which arises naturally, for example, in designing approximation algorithms by solving a linear or semidefinite program and applying some dependent randomized rounding algorithm~\cite{jacm/GandhiKPS06}. For this setting, \citet{siamcomp/PanconesiS97} showed that the Chernoff-Hoeffding bounds can be shown under the weak notion of \emph{negative cylinder dependence}: this and other standard notions of negative dependence are defined in Section~\ref{sec:neg-correl-notions}. 

For some applications in combinatorial optimization, algorithmic game theory, and machine learning, one needs to consider the more general class of \emph{submodular functions} $f$ of the random variables, rather than simple linear combinations. When the binary random variables $X_1,\ldots,X_n$ are independent, it was shown that $f(X_1,\ldots,X_n)$ still satisfies Chernoff bounds exactly~\cite{focs/ChekuriVZ10}. When there is dependence between the random variables, the results are much weaker. The only known results are for random variables that are output by specific dependent-rounding algorithms, known as swap rounding and pipage rounding~\cite{focs/ChekuriVZ10,soda/HarveyO14}. These results showed that a Chernoff-like lower-tail bound also holds for submodular functions for their specific dependent rounding procedure. As noted in the work of~\citet{garbe2018concentration}, it is not clear how to generalize either of these proofs to any general notion of negative dependence. 

We introduce a new notion of negative dependence, called 1-negative association, which is weaker than negative association and negative regression but stronger than negative cylinder dependence.
\begin{definition}
    A collection of random variables $X_1,\ldots,X_n$ is said to satisfy 1-negative association if for any two monotone functions $f$ and $g$, where $g$ depends on a single random variable $X_i$ and $f$ depends on the remaining random variables $\{X_j\}_{j\in[n]\backslash\{i\}}$, we have $\mathbb{E}[fg]\le \mathbb{E}[f]\mathbb{E}[g]$.
\end{definition} 
Importantly, while in general it is weaker than the notion of weak negative regression introduced by \citet{approx/QiuS22}, 1-negative association is equivalent to it when the variables $X_1,\ldots,X_n$ are binary. Further details are provided in Section \ref{sec:1-na}.

Our main result is that the Chernoff-like bound shown in \citet{focs/ChekuriVZ10,soda/HarveyO14} also hold under 1-negative association (see Section \ref{sec:sub}). In particular, this implies the following:
\begin{theorem}\label{thm:intro}
    Let $X_1,\ldots,X_n$ be binary random variables with mean $x_1,\ldots,x_n$ satisfying negative association (or negative regression). Let $f$ be a non-negative monotone submodular function with marginal values in $[0,1]$ and let $F$ be the multilinear extension of $f$. If we let $\mu_0=F(x_1,\ldots,x_n)$, then we have the following:
    $$\Pr[f(X_1,\ldots,X_n)\le (1-\delta)\cdot \mu_0]\le\exp(-\mu_0\delta^2/2).$$
\end{theorem}
A few remarks are in order. First, we highlight that the concentration in the above theorem is with respect to the value of the multilinear extension $F(x_1,\ldots,x_n)$, rather than the true expected value $\mathbb{E}[f(X_1,\ldots,X_n)]$. In general, the true expected value can be greater than the value of multilinear extension \cite{approx/QiuS22}. Nevertheless, this suffices for applications relating to submodular maximization, and is the same type of concentration result shown in previous work.
Second, recall that negative cylinder dependence does not suffice to show this concentration bound~\cite[p.~583]{focs/ChekuriVZ10}. As a result, our results are, in some informal sense, almost tight in terms of the condition on negative dependence. 

In addition to providing submodular concentration results for a wide class of rounding algorithms and distributions, our results also give a new path toward understanding why pipage rounding and swap rounding satisfy the lower-tail Chernoff bound. By proving that the rounding algorithms output random variables which are 1-negatively associated, we immediately obtain a new proof of the lower tail bounds. This can be viewed as evidence that the two rounding algorithms satisfy 1-negative association or even negative association/regression. We leave this as an interesting open question.

\paragraph{Techniques.} We use the standard method of bounding the exponential moments for lower-tail Chernoff bounds. Our idea is to show that the exponential moments for our negatively-correlated random variables is upper bounded by the exponential moments for independent copies of the random variables. Formally, let $X_1,\ldots,X_n$ be random variables satisfying 1-negative association and let $X_1^*,\ldots,X_n^*$ be independent copies of the random variables. We show for any $\lambda< 0$, we have
$$\mathbb{E}[\exp(\lambda\cdot f(X_1,\ldots,X_n))]\le \mathbb{E}[\exp(\lambda\cdot f(X_1^*,\ldots,X_n^*))].$$
Since the exponential-moments method has been used to prove Chernoff bounds for submodular functions in the independent case~\cite{focs/ChekuriVZ10}, we can then repeat their proof and conclude with our desired result. 
We believe this proof idea may be of independent interest. For example, the same ideas can show that for a supermodular function $g$ and any $\lambda>0$, we have
$$\mathbb{E}[\exp(\lambda\cdot g(X_1,\ldots,X_n))]\le \mathbb{E}[\exp(\lambda\cdot g(X_1^*,\ldots,X_n^*))].$$
In other words, we have morally proven the following statement: any upper-tail concentration bound which can be proven for a supermodular function $g$ under independence based on the exponential-moments method also holds when the underlying random variables are negatively associated. As an example, we can apply this to a read-$k$ family of supermodular functions $g_1,\ldots,g_r$ for negatively associated random variables~\cite{rsa/GavinskyLSS15}. A read-$k$ family is defined as a set of functions where each variable appears in at most $k$ functions. This concept is particularly useful in scenarios where functions need to model or manage overlapping sets of variables with constraints on their interaction. We highlight that the proof of concentration for read-$k$ families given in~\citet{rsa/GavinskyLSS15} doesn't use the exponential-moments method, but instead it is based on the entropy method. We address this by giving a simpler proof of their results, this time using the exponential moments method. This gives the first concentration results for a class of supermodular functions under negatively correlated random variables, and is detailed in Section \ref{sec:super}.

\paragraph{Applications.}
Our motivation for studying the problem comes from the randomized-rounding paradigm in approximation algorithms for converting a fractional solution to a linear program into an integral one. In many such randomized-rounding schemes, the output random variables have been shown to satisfy strong negative dependence properties, such as negative association \cite{Srinivasan2001DistributionsOL,jacm/GandhiKPS06}. For all such rounding algorithms, our results immediately imply the submodular Chernoff lower-tail bound. It remains an interesting open question to efficiently sample negatively dependent distributions for a wider class of set systems. A particularly interesting algorithm is given in the work of \citet{innovations/PeresSV17}; they show that a fractional point in a matroid polytope can be rounded to an integral one such that the resulting distribution preserves marginals and satisfies negative association. However, a gap identified in their proof \cite{approx/QiuS22} complicates the application of their approach. The implications of this issue for the applicability of our results remain an area for further investigation.

As a concrete application, we consider the maximum coverage problem under group fairness constraints. Here, we have a universe of elements $\{1,\ldots,n\}$, a collection $S_1,\ldots,S_m$ of subsets of the universe, and a budget $k$. We are further given subsets $C_1,\ldots,C_\ell\subseteq[n]$ (which should be thought of as demographic groups) along with thresholds $w_1,\ldots,w_\ell$. Our goal is to choose $k$ sets from the collection to maximize the number of elements covered subject to the fairness constraint each demographic group is sufficiently covered (i.e., at least $w_j$ elements from $C_j$ are covered). Since this is a special case of multiobjective submodular maximization, there exists a $(1-1/e-\epsilon)$-approximation to the problem such that each fairness constraint is approximately satisfied \cite{focs/ChekuriVZ10,nips/Udwani18}.
Unfortunately, these results rely on the randomized swap-rounding algorithm due to its submodular concentration properties, which requires a super-linear time complexity. While swap rounding can be implemented with poly-logarithmic depth \cite{DBLP:conf/soda/ChekuriQ19}, a simpler dependent-rounding algorithm of~\citet{Srinivasan2001DistributionsOL} requires linear work and only $O(\log{n})$ depth, which improves the efficiency. Observe that the pre-processing step in \citet{nips/Udwani18} only requires $O(n\ell)$ time. Since we can solve the linear program for fair maximum coverage in near-linear time \cite{allen-zhuLP}, we obtain a near-linear time algorithm for the problem after using the efficient rounding algorithm of~\citet{Srinivasan2001DistributionsOL}. These same ideas can be used to improve the time complexity of the algorithm by \citet{ijcai/TsangWRTZ19} for influence maximization with group-fairness constraints. Since the proofs are similar to previous work, we defer the details to a future version of the paper.

More generally, negatively-associated random variables show up naturally in many settings (see e.g., the primer by \citet{wajc-primer}). \citet{rsa/DubhashiR98} studied the canonical example of \emph{balls and bins}, and showed that it satisfied both negative association and negative regression. Another example satisfying the negative-association conditions are any product measure over the set of bases of a balanced matroid, as shown by \citet{stoc/FederM92}. A final setting where such random variables occur are random spanning trees, which have been vital in the recent improvements to approximation algorithms for the traveling salesperson problem (see, e.g., \cite{stoc/KarlinKG21}). Random spanning trees are known to be strongly Rayleigh, which immediately implies that they are negatively associated. Our results may be interesting here as well. 

We also observe that the online rounding scheme of \cite{naor-srinivasan-wajc:online} has the strongly Rayleigh property: we immediately get strong concentration (on the lower-tail side) for monotone submodular functions, when the inputs for the function arrive online along with their (Bernoulli) distributions as in the setup of \cite{naor-srinivasan-wajc:online}. 

\paragraph{Related Work.} 
The concentration of negatively-dependent random variables was first formally studied by \citet{newmanCLT}, which showed a central limit theorem for a certain notion of negative dependence. Later on, \citet{siamcomp/PanconesiS97} showed that cylinder negatively dependent random variables yield the Chernoff-Hoeffding concentration inequalities, just like independent random variables. In the context of our paper, these results are somewhat specialized since they focus on linear combinations of random variables. 

For non-linear functions of the random variables, the majority of work has focused on the concentration of Lipschitz functions under various notions of negative dependence. \citet{cpc/PemantleP14} showed that for strong Rayleigh measures, one has Gaussian concentration for any Lipschitz function. Later on, \citet{garbe2018concentration} corrected an earlier proof of \citet{rsa/DubhashiR98}, showing that McDiarmid-like concentration results hold for Lipschitz functions of random variables satisfying negative regression. These results are complementary to ours since we are trying to give dimension-free concentration results. 

\section{Preliminaries}

\subsection{Notions of Negative Dependence}
\label{sec:neg-correl-notions}

We begin by defining the notion of negative dependence commonly found in the literature.

\paragraph{Negative Cylinder Dependence.} A collection of Boolean random variables $X_1,\ldots,X_n$ is said to be negative cylinder dependent if for every $S \subseteq[n]$, $$ \mathbb{E}\left[\textstyle\prod_{i \in S} X_i\right] \leq \textstyle\prod_{i \in S} \mathbb{E}\left[X_i\right] $$ and $$ \mathbb{E}\left[\textstyle\prod_{i \in S}\left(1-X_i\right)\right] \leq \textstyle\prod_{i \in S} \mathbb{E}\left[1-X_i\right].$$
Negative cylinder dependence is the weaker notion considered here. It is known to imply Chernoff bounds for linear combinations of $X_1,\ldots,X_n$ but it is insufficient to show our submodular concentration results.

\paragraph{Negative Association.} A collection of random variables $X_1, \ldots, X_n$ is said to be negatively associated if for any $I, J \subset[n], I \cap J=\emptyset$ and any pair of non-decreasing functions $f:\bb R^I \rightarrow \mathbb{R}, g:\bb R^J \rightarrow \mathbb{R}$, $$ \mathbb{E}\left[f\left(X_I\right) g\left(X_J\right)\right] \leq \mathbb{E}\left[f\left(X_I\right)\right] \mathbb{E}\left[g\left(X_J\right)\right] . $$ Here and in the following, $X_S$ refers to those random variables that are indexed by the elements in 
$S$, $X_S = \{X_i\;:\;i\in S\}$. Negative association is a significant strengthening of negative cylinder dependence, and has many additional useful closure properties. This will be one of the focuses of the paper.

\paragraph{Negative Regression.} A collection of random variables $X_1, \ldots, X_n$ is said to satisfy negative regression, if for any $I, J \subset[n], I \cap J=\emptyset$, any non-decreasing function $f:\bb R^I \rightarrow \mathbb{R}$ and $a \leq b \in\bb R^J$, $$ \mathbb{E}\left[f\left(X_I\right) \mid X_J=a\right] \geq \mathbb{E}\left[f\left(X_I\right) \mid X_J=b\right].$$
Negative regression is a strengthening of negative cylinder dependence, but its relationship with negative association is not yet well understood. It is known that negative association doesn't imply negative regression \cite{rsa/DubhashiR98}, but the opposite implication is not known. This will be the other focus of the paper.

\paragraph{Strong Rayleigh.} A collection of random variables $X_1,\ldots, X_n$ is said to satisfy the strong Rayleigh property if the generating function $$F(z_1,\ldots, z_n) = \bb E[\textstyle\prod_{j=1}^nz_j^{X_j}]$$ is a real stable polynomial (i.e., it has no root $(z_1,\ldots, z_n)\in \bb C^n$ with all positive imaginary components).
The strong Rayleigh property is the strongest notion of negative dependence, and has been shown to imply all other studied negative dependence definitions \cite{strongrayleigh}. As a result, all of our results apply here as well.

\subsection{Submodular Functions}

We also give a quick review of the basics of submodular functions.

\paragraph{Submodular Functions.} We say that a function $f:\{0,1\}^n\to\mathbb{R}$ is submodular if
$$f(X_1,\ldots,X_{i-1},1,X_{i+1},\ldots X_n)-f(X_1,\ldots,X_{i-1},0,X_{i+1},\ldots X_n)$$
is a non-increasing function of $X_1,\ldots,X_{i-1},X_{i+1},\ldots,X_n$ for each $i\in[n]$. When viewing the binary input of $f$ as the indicator vector for a set, this is equivalent to the more common definition that $f$ is submodular if for any $X,Y\subseteq[n]$ with $X\subseteq Y$ and any $x\not\in Y$, we have
$$f(X\cup\{x\})-f(X)\ge f(Y\cup\{x\})-f(Y).$$

\paragraph{Supermodular Functions.} We say that a function $g:\{0,1\}^n\to\mathbb{R}$ is supermodular if
$$g(X_1,\ldots,X_{i-1},1,X_{i+1},\ldots X_n)-g(X_1,\ldots,X_{i-1},0,X_{i+1},\ldots X_n)$$
is a non-decreasing function of $X_1,\ldots,X_{i-1},X_{i+1},\ldots,X_n$ for each $i\in[n]$. When viewing the binary input of $g$ as the indicator vector for a set, this is equivalent to the more common definition that $g$ is supermodular if for any $X,Y\subseteq[n]$ with $X\subseteq Y$ and any $x\not\in Y$, we have
$$g(X\cup\{x\})-g(X)\le g(Y\cup\{x\})-g(Y).$$

\paragraph{Mutlilinear Extension.} The multilinear extension of a function $f$ is
$$F(x)=\mathbb{E}[f(x)]=\textstyle\sum_{S \subseteq N} f(S) \textstyle\prod_{i \in S} x_i \textstyle\prod_{i \notin S} (1 - x_i),$$
for $x\in[0,1]^n$. If we view $x$ as a probability vector, the multilinear extension $F$ is simply the expected value of $f$ when each coordinate is rounded independently in $\{0,1\}$.

\section{Submodular Chernoff Bounds}

\subsection{1-Negative Association and Weak Negative Regression}\label{sec:1-na}

We first define the weaker notion of negative dependence which we work with, called 1-negative association, and prove some simple properties about it. We also define a related notion of weak negative regression, which is the analogue of 1-negative association for the notion of negative regression, and we show the equivalence between the two for binary random variables and show that weak negative regression is strictly stronger in general. After an initial draft, we discovered that \citet{approx/QiuS22} had already introduced the notion of weak negative regression for binary random variables in a context complementary to ours. Using their work, we can immediately show nice properties about 1-negative association.

\begin{definition}
    A collection of random variables $X_1,\ldots,X_n$ is said to satisfy 1-negative association if for any two monotone functions $f$ and $g$, where $g$ depends on a single random variable $X_i$ and $f$ depends on the remaining random variables $\{X_j\}_{j\in[n]\backslash\{i\}}$, we have $\mathbb{E}[fg]\le \mathbb{E}[f]\mathbb{E}[g]$.
\end{definition}

\begin{definition}
    A collection of random variables $X_1,\ldots,X_n$ is said to satisfy weak negative regression if for any index $i$ and any monotone function $f$ depending on the remaining random variables $\{X_j\}_{j\in[n]\backslash\{i\}}$, we have $\mathbb{E}[f|X_i=b]\le \mathbb{E}[f|X_i=a]$ for all $a\le b$.
\end{definition}

In the following lemmata, we show that weak negative regression implies 1-negative association in general. We then show that the reverse implication holds for binary random variables, but give an example showing that it does not hold in general. 

\begin{claim}
    If a collection of random variables $X_1,\ldots,X_n$ satisfies weak negative regression, then it satisfies 1-negative association.
\end{claim}
\begin{proof}
Assume $X_1,\ldots,X_n$ satisfy weak negative regression; we will prove that it also satisfies 1-negative association. Let $f$ and $g$ be monotone functions such that $f$ depends on $X_I$ for some subset $I\subseteq[n]$ and $g$ depends on $X_i$ for $i\not\in I$. Without loss of generality, let us assume that $f$ and $g$ are non-decreasing. 


First, define the function $F(a) = \mathbb{E}[f(X) | X_i = a].$ Evaluating the expectation $\mathbb{E}[f(X_I)g(X_i)]$, we can express it via the law of total expectation as 
$$ \mathbb{E}[f(X_I)g(X_i)] = \mathbb{E}[\mathbb{E}[f(X_I)|X_i]g(X_i)] = \mathbb{E}[F(X_i)g(X_i)]. $$ 
Next, we observe that by the definition of weak negative regression, $F(a)$ is non-increasing. Therefore, random variables \(F(X_i)\) and \(g(X_i)\) are negatively correlated, yielding
$$ \mathbb{E}[F(X_i)g(X_i)] \leq \mathbb{E}[F(X_i)] \mathbb{E}[g(X_i)]. $$
Converting \(\mathbb{E}[F(X_i)]\) back into the terms of \(f\) and \(g\), we find 
$$ \mathbb{E}[F(X_i)] \mathbb{E}[g(X_i)] = \mathbb{E}[f(X)] \mathbb{E}[g(X_i)]. $$ 
The overall inequality 
 $\mathbb{E}[f(X_I)g(X_i)] \leq \mathbb{E}[f(X_I)] \mathbb{E}[g(X_i)] $ 
thus holds, which establishes that \(X_1, \ldots, X_n\) are 1-negatively associated, completing the proof.

\end{proof}

\begin{claim}
    If a collection of binary random variables $X_1,\ldots,X_n$ satisfies 1-negative association, then it satisfies weak negative regression.
\end{claim}
\begin{proof}
    Recall that we wish to prove that for any non-decreasing function $f$ depending on some subset $I\subseteq[n]$ and any $i\not\in I$, we have that
    $$\mathbb{E}[f(X_I)|X_i = 0]\geq \bb E[f(X_I)|X_i = 1].$$
    By the definition of 1-negative association, we obtain that for any monotone functions $g$ which depends on $X_i$, the following inequality holds:
    \begin{align}
        \bb E[f(X_I) g(X_i)]\leq\bb E[f(X_I)]\cdot\bb E[g(X_i)].\label{eq:1-na-condition}
    \end{align}
    Without loss of generality, we may assume that $\mathbb{E}[f(X_I)|X_i=1]=0$ by shifting $f$ by a constant. By choosing $g$ to be the identity function, we can apply the law of total probability to obtain
    \begin{align*}
            \bb E[f(X_I)\cdot g(X_i)] = \Pr [X_i = 0]\cdot \bb{E}[f(X_I)|X_i = 0]\cdot g(0)+\Pr [X_i = 1]\cdot \bb{E}[f(X_I)|X_i = 1]\cdot g(1) = 0.
    \end{align*}
    Plugging this into Equation \ref{eq:1-na-condition}, we obtain
   \begin{align*}
    0&\leq \bb E[f(X_I)]\cdot \bb E[g(X_i)]= \bb E[f(X_I)|X_i=0]\Pr[X_i=0]\cdot \mathbb{E}[X_i],
   \end{align*}
   again by the law of total probability. Since $\mathbb{E}[X_i]>0$ and $\Pr[X_i=0]>0$, this implies that $$ \bb  E[f(X_I)|X_i=0]\geq0,
    $$
    which concludes the proof since $\mathbb{E}[f(X_I)|X_i=1]=0$.
\end{proof}

\begin{claim}
There exists (non-binary) distributions over 2 random variables which satisfy 1-negative association but not weak negative regression. In other words, 1-negative association is strictly more general than weak negative regression for non-binary random variables.
\end{claim}
\begin{proof}
Let's first discuss the intuition for the construction of the counterexample. One can show via algebra that $\mathbb{E}[f(X_I)g(X_i)]-\mathbb{E}[f(X_I)]\mathbb{E}[g(X_i)]$ can be expanded as the following expression:
$$\textstyle\sum_{x<y}
\Pr[X_i=x]\Pr[X_i = y]\cdot \big(\bb E[f(X_I)|X_i = x] - \bb E[f(X_I)|X_i = y]\big)\big(g(x)-g(y)\big),
$$
where the summation is over the values that $X_i$ takes with non-zero probability.

Now, suppose that for some values $x<y$ we have that $$\bb E[f(X_1,\ldots,X_n)|X_i = x] - \bb E[f(X_1,\ldots,X_n)|X_i = y]<0,$$ violating the weak negative regression property. This would imply that some of the summands are positive (since $g(x)<g(y)$ by monotonicity). In the case of binary random variables, the summation would only consist of a single summand so 1-negative association would be violated. For general random variables, the summation consists of multiple terms so the summation may still be negative even when a single summand is positive. Consequently, the random variables may still satisfy 1-negative association. 

We now give the example. Consider the random variables $(X_1, X_2)$ which are uniformly distributed on their support set $\{(0,3),(1,1),(2,2),(3,0)\}$. By considering an identity function $\mathds 1_x:\{0,1,2,3\}\to \{0,1,2,3\}$, we can show that that the distribution of $(X_1,X_2)$ does not satisfy weak negative regression:
$$\bb E[\mathds 1_x(X_2)|X_1 = 1] = \mathds 1_x(1) = 1<2 = \mathds 1_x(2) = \bb E[\mathds 1_x(X_2)|X_1 = 2].$$
However, for any pair of non-decreasing functions $f,g:\{0,1,2,3\}\to \bb R$, we have
$$
\bb E[f(X_1)]\bb E[g(X_2)] - \bb E[f(X_1)g(X_2)] = \frac{f(1)+f(2)+f(3)}{4}\cdot\frac{g(1)+g(2)+g(3)}{4} - \frac{f(1)g(1)+f(2)g(2)}{4},
$$
where we have again assumed without loss of generality that $f(0)=g(0)=0$. 

We claim that the quantity on the right hand side is always non-negative. In order to see this, observe that $f(2)g(2)\le f(i)g(j)$ for any $i,j\ge 2$ by monotonicity. As a result, we have
$$\frac{f(2)g(2)}{4}\leq \frac{(f(2)+f(3))(g(2)+g(3))}{16}.$$
Further, we observe that $f(1)g(1)\le f(i)g(j)$ for any $i,j\ge 1$ by monotonicity. As a result, we have
$$\frac{f(1)g(1)}{4}\leq \frac{f(1)(g(2)+g(3)) + g(1)(f(2)+f(3))}{16}.$$
Combining these two inequalities immediately and observing that $f(1)g(1)\ge 0$ by monotonicity implies our desired result. Hence, the distribution is 1-negatively associated.
\end{proof}

Since 1-negative association and weak negative regression are equivalent for binary random variables and weak negative regression has been shown to be strictly stronger than cylinder negative dependence \cite[Proposition 2.4]{approx/QiuS22}, we also have that $1$-negative association is strictly stronger than cylinder negative dependence. Additionally, since weak negative regression is strictly stronger than $1$-negative association for general random variables and weak negative regression has been shown to be strictly weaker than negative association and negative regression \cite[Proposition 2.4]{approx/QiuS22}, we have that 1-negative association is strictly weaker than negative association and negative regression. We summarize these in the following corollaries.

\begin{corollary}
    1-negative association is a strictly weaker condition than negative association.\label{lem:1-na}
\end{corollary}

\begin{corollary}
    1-negative association is a strictly weaker condition than negative regression.\label{lem:1-nr}
\end{corollary}

\begin{corollary}
    1-negative association is a strictly stronger condition than negative cylinder dependence.
\end{corollary}

\subsection{Proof of Submodular Concentration}\label{sec:sub}

We will now prove our main result. As mentioned in the introduction, our proof is based on the standard technique of bounding the exponential moments. The following lemma contains our main technical contribution, stating that the exponential moments of $f(X_1,\ldots,X_n)$ under 1-negative association is dominated by that under independence. Our results will follow easily afterwards.

\begin{lemma}\label{lem:exp-mom}
    Let $X_1,\ldots,X_n$ be 1-negatively associated random variables and let $X_1^*,\ldots,X_n^*$ be independent random variables with the same marginal distributions. Also let $f$ be a non-negative monotone function.
    \begin{itemize}
        \item If $f$ is a submodular function and $\lambda<0$, we have $\mathbb{E}[\exp(\lambda f(X_1,\ldots,X_n))]\le \mathbb{E}[\exp(\lambda f(X_1^*,\ldots,X_n^*))]$.
        \item If $f$ is a supermodular function and $\lambda>0$, we have $\mathbb{E}[\exp(\lambda f(X_1,\ldots,X_n))]\le \mathbb{E}[\exp(\lambda f(X_1^*,\ldots,X_n^*))]$.
    \end{itemize}
\end{lemma}
\begin{proof}
Fix $\lambda<0$ if $f$ is submodular and $\lambda>0$ if $f$ is supermodular. Observe that in order to prove the lemma, it suffices to prove
\begin{align}
    \mathbb{E}[\exp(\lambda\cdot f(X_1,\ldots,X_i,\ldots,X_m))]\le \mathbb{E}[\exp(\lambda\cdot f(X_1,\ldots,X_i^*,\ldots,X_m))],\label{eq:one-step}
\end{align}
since we can iteratively apply the above inequality to each $X_i$ (note that we can do this because independent variables are also negatively associated). For simplicity of notation and without loss of generality, we will prove the inequality for $i=1$.

By considering the cases of $X_1=0$ and $X_1=1$ separately, we have 
\begin{align*}
    \exp(\lambda\cdot f(X_1,\ldots,X_n))&=X_1\cdot\exp(\lambda\cdot f(1,X_2,\ldots,X_n)))+(1-X_1)\cdot \exp(\lambda\cdot f(0,X_2,\ldots,X_n))),
\end{align*}
where the equality holds pointwise on the underlying probability space. Via simple algebraic manipulations, we can rewrite the right hand side as
\begin{align*}
    &X_1\cdot\left[\exp(\lambda\cdot f(1,X_2,\ldots,X_n)))-\exp(\lambda\cdot f(0,X_2,\ldots,X_n))\right]+\exp(\lambda\cdot f(0,X_2,\ldots,X_n))
\end{align*}
Taking expectations, we now have that $\mathbb{E}[\exp(\lambda f(X_1,\ldots,X_n))]$ can be written as
\begin{align}
    \mathbb{E}\big[X_1\cdot\left[\exp(\lambda\cdot f(1,X_2,\ldots,X_n)))-\exp(\lambda\cdot f(0,X_2,\ldots,X_n))\right]\big]+\mathbb{E}\big[\exp(\lambda\cdot f(0,X_2,\ldots,X_n))\big]\label{eq:sum}
\end{align}
Observe that $X_1$ is clearly an increasing function of $X_1$. We claim that if either ($i$) $f$ is submodular and $\lambda<0$ or ($ii$) $f$ is supermodular and $\lambda>0$, we have that $\exp(\lambda\cdot f(1,X_2,\ldots,X_n)))-\exp(\lambda\cdot f(0,X_2,\ldots,X_n))$ is an increasing function in $X_2,\ldots, X_n$. Indeed, we first rewrite the function as 
\begin{align*}
    \exp(\lambda\cdot f(0,X_2,\ldots,X_n))\cdot [\exp(\lambda\cdot(f(1,X_2,\ldots,X_n)-f(0,X_2,\ldots,X_n)))-1]\eqqcolon A_1\cdot A_2
\end{align*}
for simplicity of notation.

Let us first consider the case when $\lambda<0$ and $f$ is submodular. We have that $A_1$ is ($i$) positive because the exponential function is always positive and ($ii$) non-increasing in $X_2,\ldots, X_n$ because $f$ is non-decreasing and $\lambda<0$. We also have that $A_2$ is ($i$) negative because the argument in $\exp(\cdot)$ is negative, so the exponential is in $(0,1)$ ($ii$) non-decreasing since $\lambda<0$ and the difference of $f$ evaluated at $X_1=1$ and $X_1=0$ is non-increasing by definition of submodularity.
Hence, our expression of interest is the product of a function $A_1$ which decreases towards 0 and a function $A_2$ which increases towards $0$. The product will be negative and monotonically increasing towards $0$.

Now, let us consider the case when $\lambda>0$ and $f$ is supermodular. We have that $A_1$ is ($i$) positive because the exponential function is always positive and ($ii$) non-decreasing since $\lambda>0$, $f$ is monotone, and $\exp(\cdot)$ is also monotone. We also have that $A_2$ is ($i$) positive because the argument of $\exp(\cdot)$ is positive since $f$ is monotone so the exponential is greater than $1$ and ($ii$) non-decreasing since $\lambda>0$ and the difference of $f$ evaluated at $X_1=1$ and $X_1=0$ is non-decreasing by definition of supermodularity. As a result, the product will be positive and non-decreasing, as desired.

Since we have shown that the $A_1A_2$ is also monotone, we now have that the first term in Equation \ref{eq:sum} can be written as the product of monotone functions of disjoint subsets, one of which is the singleton set. By 1-negative association, we have that the first term is upper bounded by
$$\mathbb{E}[X_1]\cdot\mathbb{E}[\exp(\lambda\cdot f(1,X_2,\ldots,X_n))-\exp(\lambda\cdot f(0,X_2,\ldots,X_n))].$$
Consequently, the entire expression in (\ref{eq:sum}) is upper bounded by
$$\mathbb{E}[X_1]\cdot\mathbb{E}[\exp(\lambda\cdot f(1,X_2,\ldots,X_n))-\exp(\lambda\cdot f(0,X_2,\ldots,X_n))]+\mathbb{E}[\exp(\lambda\cdot f(0,X_2,\ldots,X_n))].$$
Since $X_1$ and $X_1^*$ have the same marginal distributions, the above is exactly equal to
$$\mathbb{E}[X_1^*]\cdot\mathbb{E}[\exp(\lambda\cdot f(1,X_2,\ldots,X_n))-\exp(\lambda\cdot f(0,X_2,\ldots,X_n))]+\mathbb{E}[\exp(\lambda\cdot f(0,X_2,\ldots,X_n))].$$
And since $X_1^*$ is independent with $X_2,\ldots,X_m$ by assumption, the above is equal to
$$\mathbb{E}[X_1^*\cdot\exp(\lambda\cdot f(1,X_2,\ldots,X_n))-\exp(\lambda\cdot f(0,X_2,\ldots,X_n))]+\mathbb{E}[\exp(\lambda\cdot f(0,X_2,\ldots,X_n))].$$
In particular, observe that this is in the exact same form as Equation \ref{eq:sum}, except with $X_1$ replaced with $X_1^*$. Note that when we transformed the left-hand side of Equation \ref{eq:one-step} to Equation \ref{eq:sum}, we never used any properties of the random variables $X_1,\ldots, X_n$ other than the fact that they take values in $\{0,1\}$. As a result, we can reverse the direction of all of the equalities to show that the above expression is equal to 
$$\mathbb{E}[\exp(\lambda\cdot f(X_1^*,X_2,\ldots,X_n))],$$
which completes the proof of the lemma.
\end{proof}

Now, we will complete the proof of our main result. Combining the theorem below with Claims \ref{lem:1-na} and \ref{lem:1-nr} immediately gives a proof of Theorem \ref{thm:intro}. Here, our proof will rely heavily on the proof of the Chernoff bound for submodular functions under independence given in \citet{focs/ChekuriVZ10}.

\begin{theorem}
    Let $X_1,\ldots,X_n$ be binary random variables with mean $x_1,\ldots,x_n$ satisfying 1-negative association. Let $f$ be a non-negative monotone submodular function with marginal values in $[0,1]$ and let $F$ be the multilinear extension of $f$. If we let, $\mu_0=F(x_1,\ldots,x_n)$, then we have the following:
    $$\Pr[f(X_1,\ldots,X_n)\le (1-\delta)\cdot \mu_0]\le\exp(-\mu_0\delta^2/2).$$\label{thm:main-concentration}
\end{theorem}
\begin{proof}
    Let $X_1^*,\ldots,X_n^*$ be independent random variables with the same respective marginals as $X_1,\ldots,X_n$ and let $\lambda<0$ be a parameter to be set later. Let us decompose $f(X_1^*,\ldots,X_n^*)=\sum_{i=1}^{n}Y_i^*$, where
    $$Y_i^*=f(X_1^*,\ldots,X_i^*,0,\ldots,0)-f(X_1^*,\ldots,X_{i-1}^*,0,\ldots,0).$$
    Let us denote $\mathbb{E}[Y_i^*]=\omega_i$ and $\mu_0=\sum_{i=1}^{n}\omega_i=\mathbb{E}[f(X_1^*,\ldots,X_n^*)]$. By the convexity of the exponential and the fact that $Y_i^*\in[0,1]$, we have that
    $$\mathbb{E}[\exp(\lambda\cdot Y_i^*)]\le \omega_i\cdot\exp(\lambda)+(1-\omega_i)=1+[\exp(\lambda)-1]\cdot\omega_i\le \exp[(\exp(\lambda)-1)\cdot\omega_i].$$
    Combining the above with Lemma C.1 from \citet{focs/ChekuriVZ10}, we have that
    \begin{align}
        \mathbb{E}[\exp(\lambda\cdot f(X_1^*,\ldots,X_n^*)]=\mathbb{E}[\exp(\lambda\cdot\textstyle\sum_{i=1}^{n}Y_i^*)]\le \textstyle\prod_{i=1}^{n}\mathbb{E}[\exp(\lambda\cdot Y_i^*)]\le \exp[(\exp(\lambda)-1)\cdot \mu_0].\label{eq:calc}
    \end{align}
    Now, we can follow the proof of the standard Chernoff bound:
    \begin{align*}
        \Pr[f(X_1,\ldots,X_n)\le(1-\delta)\cdot\mu_0]&=\Pr[\exp(\lambda\cdot f(X_1,\ldots,X_n))\le\textcolor{red}{\ge} \exp(\lambda(1-\delta)\cdot\mu_0)]\\
        &\le \frac{\mathbb{E}[\exp(\lambda\cdot f(X_1,\ldots,X_n))]}{\exp(\lambda(1-\delta)\cdot\mu_0)}\\
        &\le \frac{\mathbb{E}[\exp(\lambda\cdot f(X_1^*,\ldots,X_n^*))]}{\exp(\lambda(1-\delta)\cdot\mu_0)}\\
        &\le \frac{\exp[(\exp(\lambda)-1)\cdot\mu_0]}{\exp(\lambda(1-\delta)\cdot\mu_0)}
    \end{align*}
    The first equality follows since $\exp(\lambda\cdot x)$ is a monotone function, the first inequality follows by Markov's inequality, the second inequality follows by Lemma \ref{lem:1-na}, and the final inequality follows Equation \ref{eq:calc}.

    Finally, we can choose $\lambda$ such that $\exp(\lambda)=1-\delta$, which gives
    $$\Pr[f(X_1,\ldots,X_n)\ge (1-\delta)\mu_0]\le \frac{\exp(-\delta\mu_0)}{(1-\delta)^{(1-\delta)\mu_0}}\le \exp(-\mu_0\cdot \delta^2/2),$$
    where we used $(1-\delta)^{1-\delta}\le \exp(-\delta+\delta^2/2)$ for $\delta\in(0,1]$ in the final inequality.
\end{proof}

\subsection{Concentration of read-$k$ families}\label{sec:super}

In this subsection, we illustrate an application of our proof technique to give concentration for a read-$k$ family of supermodular functions. Read-$k$ families arise naturally in problems such as subgraph counting in random graphs, and can be seen as a complementary weak dependence notion to that of low-degree polynomials~\cite{combinatorica/KimV00}. Our work gives the first concentration results for these problems under negative dependence.

Let's consider this notion of weak dependence defined in \citet{rsa/GavinskyLSS15}. Let $Y_1,\ldots,Y_n$ be random variables and assume that they can be factored as functions of random variables $X_1,\ldots,X_m$. We say that $Y_1,\ldots,Y_n$ are a read-$k$ family of $X_1,\ldots,X_m$ if for each variable $X_i$, there are at most $k$ variables among $Y_1,\ldots, Y_n$ that are influenced by $X_i$. Formally, we have the following.

\begin{definition}
    Let $X_1,\ldots,X_m$ be random variables. For each $j\in[n]$, let $P_j\subseteq[m]$ and let $f_j:\{0,1\}^{P_j}\to [0,1]$ be functions of $X_{P_j}$. We say that $Y_j=f_j(X_{P_j})$ are a read-$k$ family if $|\{j:i\in P_j\}|\le k$ for each $i\in[m]$ (i.e., each variable $X_i$ influences at most $k$ functions).
\end{definition}

When $X_1,\ldots,X_m$ are independent, \citet{rsa/GavinskyLSS15} showed that we have
\begin{align}\label{eq: gav+}
&\Pr[\textstyle\sum_{j=1}^{n}f_j(X_{P_j})\ge (p+\epsilon)n]\le\exp(-D(p+\epsilon||p)\cdot n/k)
\\
\label{eq: gav-}
&\Pr[\textstyle\sum_{j=1}^{n}f_j(X_{P_j})\le (p-\epsilon)n]\le\exp(-D(p-\epsilon||p)\cdot n/k),
\end{align}
where $p=(1/n)\sum_{j=1}^{n}\mathbb{E}[Y_j]$ and $D(\cdot||\cdot)$ is the Kullback-Leibler divergence. Notably, while Gavinsky et al. do not require random variables $X_1,\ldots,X_m$ to be binary, as we do in our approach, their model requires $Y_1,\ldots,Y_n$ to be binary. We will show that Inequality \ref{eq: gav+} on the upper tail still holds for supermodular functions $f_1,\ldots,f_n$. Similarily, Inequality \ref{eq: gav-}, which addresses the lower tail bound, continues to apply to submodular functions $f_1,\ldots,f_n$.
\begin{theorem}\label{thm: supermodular}
    Let $X_1,\ldots,X_m$ be $1$-negatively associated random variables and let $X_1^*,\ldots,X_m^*$ be independent random variables with the same respective marginal distributions as $X_1,\ldots,X_m$. Suppose that $f_j(X_{P_j})$ for $j\in[n]$ are a read-$k$ family, where $f_j:\{0,1\}^{P_j}\to [0,1]$ are supermodular functions. If we let $p_0=(1/n)\sum_{j=1}^{n}\mathbb{E}[f_j(X^*_{P_j})]$ denote the averaged expectation when the underlying random variables are independent, we have
    $$\Pr[\textstyle\sum_{j=1}^{n}f_j(X_{P_j})\ge (p_0+\epsilon)n]\le\exp(-D(p_0+\epsilon||p_0)\cdot n/k).$$
\end{theorem}
\begin{proof}
    Let $f(X_1,\ldots,X_m)=\sum_{j=1}^{n}f_j(X_{P_j})$ be the quantity of interest, and note that $f$ is the sum of supermodular functions so it is supermodular as well.
    
    We will follow the standard proof via exponential moments. Let $\lambda>0$; we have
    \begin{align}
        \Pr[f(X_1,\ldots,X_m)\ge (p_0+\epsilon)n]&=\Pr[\exp(\lambda\cdot f(X_1,\ldots,X_m))\ge \exp(\lambda\cdot (p_0-\epsilon)n)]\label{eq:start}\\
        &\le \mathbb{E}[\exp(\lambda\cdot f(X_1,\ldots,X_m))]/\exp(\lambda\cdot(p_0-\epsilon)n),
    \end{align}
    where the inequality follows by Markov's. Since $f$ is supermodular, we have by Lemma \ref{lem:exp-mom} that
    \begin{align}
        \mathbb{E}[\exp(\lambda\cdot f(X_1,\ldots,X_m))]\le \mathbb{E}[\exp(\lambda\cdot f(X_1^*,\ldots,X_m^*))]= \mathbb{E}[\exp(\lambda\cdot\textstyle\sum_{j=1}^{n}f_j(X_{P_j}^*))].
    \end{align}
In the new proof of concentration of read-$k$ families, given in the Appendix, we show that
\begin{align}
    \mathbb{E}[\exp(\lambda\cdot\textstyle\sum_{j=1}^{n}f_j(X_{P_j}^*))]\le\left(\textstyle\prod_{j=1}^{n}\mathbb{E}[\exp(\lambda\cdot f_j(X^*_{P_j}))^k]\right)^{1/k}.\label{eq:end}
\end{align}
Combining equations \ref{eq:start}--\ref{eq:end}, we have
$$\Pr[f(X_1,\ldots,X_m)\ge(p_0+\epsilon)n]\le \left(\textstyle\prod_{j=1}^{n}\mathbb{E}[\exp(k\lambda\cdot f_j(X^*_{P_j}))/\exp(k\lambda(p_0+\epsilon)n]\right)^{1/k}.$$
Let $\lambda^\prime=k\lambda$; since $\lambda>0$ is a parameter we set, we can view $\lambda^\prime>0$ as a parameter as well. We will abuse notation and replace $\lambda^\prime$ with $\lambda$, so we have
$$\Pr[f(X_1,\ldots,X_m)\ge(p_0+\epsilon)n]\le \left(\textstyle\prod_{j=1}^{n}\mathbb{E}[\exp(\lambda\cdot f_j(X^*_{P_j}))/\exp(\lambda(p_0+\epsilon)n]\right)^{1/k},$$
for any $\lambda>0$.
Now, observe that the right hand side of the inequality is the exact same as in the proof of the standard Chernoff bound under independence, except with an additional exponent $1/k$. As a result, we can follow the original proof of the Chernoff bound to show that
$$\Pr[\textstyle\sum_{j=1}^{n}f_j(X_{P_j})\ge (p_0+\epsilon)n]\le\exp(-D(p_0+\epsilon||p_0)\cdot n/k),$$
which was our desired result.
\end{proof}

\begin{corollary}
    Let $X_1,\ldots,X_m$ be $1$-negatively associated random variables. Suppose that $f_j(X_{P_j})$ for $j\in[n]$ are a read-$k$ family, where $f_j:\{0,1\}^{P_j}\to [0,1]$ are submodular functions. If we let $p_0=(1/n)\sum_{j=1}^{n}\mathbb{E}[f_j(X^*_{P_j})]$ denote the averaged expectation when the underlying random variables are independent, we have
    $$\Pr[\textstyle\sum_{j=1}^{n}f_j(X_{P_j})\le (p_0-\epsilon)n]\le\exp(-D(p_0-\epsilon||p_0)\cdot n/k).$$
\end{corollary}
\begin{proof}
    Define $g_j:=1 - f_j$, where $g_j:\{0,1\}^{P_j}\to [0,1]$ are supermodular. Then $1 - p_0 = (1/n)\sum_{j=1}^{n}\mathbb{E}[g_j(X^*_{P_j})]$. Applying Theorem \ref{thm: supermodular}, we have:
    $$
    \Pr[\textstyle\sum_{j=1}^{n}g_j(X_{P_j})\ge (1-(p_0-\epsilon))n] \le
    (1-  p_0+\epsilon)n]\le\exp(-D(1-p_0+\epsilon||1 - p_0)\cdot n/k).
    $$
    The property of Kullback–Leibler divergence $D(1-p||1-q) = D(p||q)$ implies
    \begin{align*}
        \Pr[\textstyle\sum_{j=1}^{n}f_j(X_{P_j})\le (p_0-\epsilon)n] &= \Pr[\textstyle\sum_{j=1}^{n}g_j(X_{P_j})\ge (1-(p_0-\epsilon))n] \\
        &\le\exp(-D(p_0-\epsilon||p_0)\cdot n/k),
    \end{align*}
    which concludes the proof.
\end{proof}

\appendix
\section{Concentration of Read-$k$ Families}

We will give a new simpler proof of the results of Gavinsky et al.~\cite{rsa/GavinskyLSS15} using exponential moments for $f_j$ for independent random variables $X_1,\ldots,X_n$. Our proof will use the following lemma 
\begin{lemma}
    For an arbitrary read-$k$ family $F_1,\ldots,F_n$, we have
    \begin{align}
\mathbb{E}\left[\textstyle\prod_{j=1}^{n}F_j\right]\le \textstyle\left(\prod_{j=1}^{n}\mathbb{E}\left[F_j^k\right]\right)^{1/k}.\label{eq:1k-CH}
    \end{align}
\end{lemma}
Using this lemma, if we define 
$$F_j=\exp(\lambda\cdot f_j(X_1,\ldots,X_m)),$$
it is easy to see that $F_1,\ldots,F_n$ are a read-$k$ family since $f_1,\ldots,f_n$ are read-$k$. We will show later that after applying inequality (\ref{eq:1k-CH}) on $F_1,\ldots,F_m$, we can adapt the standard Chernoff bound proof to our case.

\begin{proof}
    We will prove inequality (\ref{eq:1k-CH}) via induction on the number of independent variables $m$. For the base case $m=1$, observe that there are at most $k$ non-constant functions $F_j$ by definition. If there are fewer than $k$ functions $F_j$, we can also add identity functions without changing the product. As a result, we can without loss of generality assume that there are exactly $k$ functions $F_j$. The inequality then follows directly by the Generalized H\"older's Inequality.

Now assume we have proven the statement for $m$ independent variables; we will try to prove it for $m+1$. Again, let $S=S_1$ denote the set of functions $F_j$ which are influenced by $X_1$. We have
\begin{align}             
    \mathbb{E}\left[\textstyle\prod_{j=1}^{n}F_j\right]&=\mathbb{E}\left[\textstyle\prod_{j\in S}F_j\cdot\prod_{j\not\in S}F_j\right]\nonumber\\
    &= \mathbb{E}\left[\mathbb{E}\left[\textstyle\prod_{j\in S}F_j\cdot\prod_{j\not\in S}F_j\big|X_2,\ldots, X_{m+1}\right]\right]\nonumber\\
    &= \mathbb{E}\left[\mathbb{E}\left[\textstyle\prod_{j\in S}F_j\big|X_2,\ldots, X_{m+1}\right]\cdot\textstyle\prod_{j\not\in S}F_j\right]\label{eq:idk}
\end{align}
Here, the first equality is obvious, the second equality follows by the law of total expectation, and the third equality follows since $F_j$ for $j\not\in S$ only depends on $X_2,\ldots, X_{m+1}$ and is independent of $X_1$. 

After taking the conditional expectation, observe that $\mathbb{E}[\prod_{j\in S}F_j|X_2,\ldots,X_{m+1}]$ is a random variable which only depends on $X_2,\ldots,X_{m+1}$. In particular, these form a read-$k$ family over $m$ random variables, so we can apply the inductive hypothesis to claim that
\begin{align*}
    \textstyle\mathbb{E}\left[\prod_{j\in S}F_j|X_2,\ldots,X_{m+1}\right]\le \left(\prod_{j\in S}\mathbb{E}\left[F_j^k|X_2,\ldots,X_{m+1}\right]\right)^{1/k}
\end{align*}
After combining this with equation (\ref{eq:idk}), we have
\begin{align*}
\textstyle\mathbb{E}\left[\prod_{j=1}^{n}F_j\right]\le \mathbb{E}\left[\prod_{j\in S}\mathbb{E}\left[F_j^k|X_2,\ldots,X_{m+1}\right]^{1/k}\cdot\prod_{j\not\in S}F_j\right]\eqqcolon A
\end{align*}
For $j\not\in S$, let $G_j=F_j$ and for $j\in S$, define 
$$G_j=\mathbb{E}\left[F_j^k|X_2,\ldots,X_{m+1}\right]^{1/k}$$
so that $A=\mathbb{E}\left[\prod_{j=1}^{n}G_j\right]$.
Observe that $G_j$ are again a read-$k$ family on $X_2,\ldots,X_{m+1}$, so we can again apply the induction hypothesis to claim that 
\begin{align*}
A=\mathbb{E}\left[\textstyle\prod_{j=1}^{n}G_j\right]&\le \textstyle\prod_{j=1}^{n}\mathbb{E}\left[G_j^k\right]^{1/k}=\textstyle\prod_{j=1}^{n}\mathbb{E}\left[F_j^k\right]^{1/k},
\end{align*}
where the final equality follows directly by definition of $G_j$. This completes the proof.
\end{proof}

Using the lemma, we can follow the standard Chernoff bound techniques to complete the proof. Let $X=\sum_{j=1}^{n}f_j(X_1,\ldots,X_m)$; we can apply Markov's inequality for any $\lambda<0$ to obtain
$$\Pr[X\le qn]=\Pr[\exp(\lambda X)\ge\exp(\lambda qn)]\le \exp(-\lambda qn)\cdot\mathbb{E}\left[\exp(\lambda X)\right]$$
As mentioned before, we can take 
$F_j=\exp(\lambda\cdot f_j(X_1,\ldots,X_m))$ and apply inequality (\ref{eq:1k-CH}) to obtain that
$$\exp(\lambda X)=\mathbb{E}\left[\textstyle\prod_{j=1}^{n}F_j\right]\le\textstyle\prod_{j=1}^{n}\mathbb{E}[F_j^k]^{1/k}.$$
Combining this inequality with the previous inequality, we now have that
$$\Pr[X\le qn]\le \exp(-\lambda qn)\cdot\textstyle\prod_{j=1}^{n}\mathbb{E}[F_j^k]^{1/k}=\left(\textstyle\prod_{j=1}^{n}\mathbb{E}\left[[F_j/\exp(\lambda q)]^k\right]\right)^{1/k}.$$
Writing out the definition of $F_j$, we have 
$$[F_j/\exp(\lambda q)]^k=\exp(\lambda k\cdot f_j(X_1,\ldots,X_m))/\exp(\lambda kq).$$
Define $\lambda^\prime=\lambda k$; since $\lambda$ was arbitrary, we will abuse notation and let $\lambda=\lambda^\prime$. Combining the two previous inequalities, we have that
$$\Pr[X\le qn]\le\left(\textstyle\prod_{j=1}^{n}\mathbb{E}[\exp(\lambda\cdot f_j)/\exp(\lambda q)]\right)^{1/k}.$$
Here, the right-hand side is in exactly the same form as in the proof of the Chernoff-Hoeffding theorem under independence, except with an additional exponent $1/k$. As a result, we can follow the Chernoff-Hoeffding proof and take $q=p-\epsilon$ to obtain
$$\Pr[X\le (p-\epsilon)n]\le \exp(-D(p-\epsilon||p)\cdot n/k),$$
which was our desired result.

\printbibliography
%
%
\end{document}